\documentclass[journal,twoside,web]{IEEEcolor}

\usepackage{generic}
\usepackage[sort,noadjust]{cite}
\usepackage{amsmath,amssymb,amsfonts}
\usepackage{algorithmic}
\usepackage{graphicx}
\usepackage{algorithm,algorithmic}
\usepackage{hyperref}
\hypersetup{hidelinks=true}
\usepackage{textcomp}
\def\BibTeX{{\rm B\kern-.05em{\sc i\kern-.025em b}\kern-.08em
    T\kern-.1667em\lower.7ex\hbox{E}\kern-.125emX}}
\usepackage[english]{babel}
\usepackage[utf8]{inputenc}
\usepackage{amsfonts}
\usepackage{amssymb}
\usepackage{mathtools}
\usepackage{color}
\usepackage{hyperref}
\usepackage{listings}
\usepackage{setspace}
\usepackage{nicefrac}
\usepackage{csquotes}
\usepackage{calc}
\usepackage{url}
\usepackage{longtable}
\usepackage{stackrel}
\usepackage{ mathrsfs }
\usepackage{matlab-prettifier}
\usepackage{soul}
\let\labelindent\relax
\usepackage{enumitem}
\usepackage{eqparbox}
\usepackage{tabularx}
\usepackage{bbm}
\usepackage{extarrows}
\usepackage[normalem]{ulem}
\usepackage[font=small]{caption}
\usepackage[official]{eurosym}
\usepackage{booktabs}
\usepackage{lipsum}
\usepackage{tikz}
\newcommand{\R}{\mathbb{R}}

\newcommand{\B}{\mathbb{B}}

\newcommand{\E}{\mathbb{E}}
\renewcommand{\P}{\mathbb{P}}

\newcommand{\td}{\text{\normalfont d}}

\newcommand{\n}[1]{\left| #1\right|}

\newcommand{\eqb}{\begin{flalign}}
\newcommand{\eqe}{\end{flalign}}

\newcommand{\mc}{\mathcal}
\newcommand{\mf}{\mathfrak}

\newcommand{\tmc}[1]{\widetilde{\mc #1}}

\newcommand{\omitcustom}[1]{{}}

\newcommand{\wt}{\widetilde}

\newcommand\numberthis{\addtocounter{equation}{1}\tag{\theequation}}

\newcommand{\del}{\partial}

\allowdisplaybreaks

\graphicspath{{images/}}

\usepackage{theoremref}

\usepackage{amsthm}
\theoremstyle{theorem}
\newtheorem{thm}{Theorem} 
\newtheorem{prop}{Proposition} 
\newtheorem{corollary}{Corollary} 
\newtheorem{remark}{Remark} 
\newtheorem{assm}{Assumption}

\newtheorem{defn}{Definition} 
\newtheoremstyle{noperiod}
  {\topsep}   
  {\topsep}   
  {\normalfont}  
  {0pt}       
  {\bfseries} 
  {}          
  {5pt plus 1pt minus 1pt} 
  {}          
\theoremstyle{noperiod}

\newlength{\leftstackrelawd}
\newlength{\leftstackrelbwd}
\def\leftstackrel#1#2{\settowidth{\leftstackrelawd}%
{${{}^{#1}}$}\settowidth{\leftstackrelbwd}{$#2$}%
\addtolength{\leftstackrelawd}{-\leftstackrelbwd}%
\leavevmode\ifthenelse{\lengthtest{\leftstackrelawd>0pt}}%
{\kern-.5\leftstackrelawd}{}\mathrel{\mathop{#2}\limits^{#1}}}

\begin{document}

\title{On robustness, input-to-state stability and backstepping for stochastic differential equations}
\author{Robert H. Moldenhauer, Dragan Nešić \IEEEmembership{Fellow, IEEE},  Mathieu Granzotto, Romain Postoyan \IEEEmembership{Senior Member, IEEE}, and Andrew R. Teel  \IEEEmembership{Fellow, IEEE}
\thanks{
This work was supported by the ARC under the Discovery Project DP210102600, by the ANR grant OLYMPIA ANR-23-CE48-0006, the AFOSR grant FA9550-25-1-0186 and the
ARO grant W911NF-26-1-0001.}
\thanks{R. Moldenhauer, D. Nešić and M. Granzotto are with the Department of Electrical and Electronic Engineering, University of Melbourne, Parkville, VIC 3010, Australia (e-mail: moldenhauer.r@student.unimelb.edu.au, mathieu.granzotto@unimelb.edu.au, dnesic@unimelb.edu.au).}
\thanks{R. Moldenhauer and R. Postoyan are with the Université de Lorraine, CNRS, CRAN, F-54000 Nancy, France (emails: \{name.surname\}@univ-lorraine.fr).}
\thanks{A. Teel is with the Electrical and Computer Engineering Department, University of California, Santa Barbara, CA 93106 USA (e-mail: teel@ece.ucsb.edu).}
}
\maketitle

\begin{abstract}
    We study conditions under which stability of the origin of stochastic differential equations is robust to small perturbations.
    We express robustness in two ways, firstly in the sense that stochastic stability is maintained under small parametric perturbations not exceeding a state-dependent bound vanishing at the origin but positive elsewhere, and secondly via stochastic input-to-state stability (ISS) which allows non-zero perturbations everywhere.
    We prove the former property assuming the existence of a Lyapunov function certifying stochastic stability of the nominal system.
    Under the same assumption, stochastic ISS holds under a suitable state-dependent perturbation scaling.
    Stochastic exponential stability is maintained under proportionally bounded perturbations and implies exponential ISS even without perturbation scaling.
    Finally, we propose a novel approach to stochastic integrator backstepping in pure-feedback form that uses the tools from our robustness analysis.
\end{abstract}


\begin{center}
    {\small This work has been submitted to the IEEE for possible publication.
    Copyright may be transferred without notice, after which this
    version may no longer be accessible.}
\end{center}

\section{Introduction}
Robust stability is a property of dynamical systems that characterizes stability under perturbations and uncertainty.
One way to express robustness is that stability is maintained under perturbations not exceeding a state-dependent bound 
vanishing at the equilibrium and positive everywhere else.
For ordinary differential equations, it was shown in \cite{lin1996smooth}, using a smooth converse Lyapunov function theorem, that global asymptotic stability (GAS) of the origin exhibits this form of robustness under mild regularity conditions.
An alternative way to express robustness is through input-to-state stability (ISS), which allows non-vanishing perturbations at the equilibrium.
In \cite{sontag_further_1990} it was shown that GAS of the origin for the nominal (unperturbed) system implies ISS when the input is scaled with a state-dependent function, which is referred to as \textit{weak ISS}\footnote{To be precise, \cite[Theorem 1]{sontag_further_1990} states that \emph{stabilizability} implies weak input-to-state \emph{stabilizability}. However, this is proved by reduction to the case where the stabilizing feedback policy is zero, and no change of this policy is required to achieve weak ISS.
This is in contrast to \cite{sontag1989smooth} where, for control affine systems, it was shown that stabilizability implies input-to-state stabilizability without the need of state-dependent scaling, but generally requiring a change of feedback policy, as nominal GAS does not imply ISS.}.

\omitcustom{
Backstepping is a recursive nonlinear control design procedure for interconnected (cascaded?) systems, see, e.g., \cite{krstic1995nonlinear} (I am struggling with this sentence).
Sontag's result on weak ISS \cite{sontag_further_1990} can be used for integrator (also more general?) backstepping (cite who did that).
The idea is to exploit that the final system in the cascade, if stabilizable, is ISS under a state-dependent input scaling.
By choosing a control law that renders the dynamics of the scaled input variable GAS, a cascade of a GAS and ISS system is obtained, which is GAS as a special case of the ISS small-gain theorem \cite{jiang_small-gain_1994}.
}

The goal of this paper is to study these robustness concepts for stochastic differential equations (SDEs) and apply them to backstepping.
For (discrete-time) stochastic difference inclusions, a very well-rounded treatment of converse theorems, robustness and ISS was achieved in \cite{teel_equivalent_2014, teel_matrosov_2013, subbaraman_converse_2013}.
There, smooth converse theorems and robustness are proved for the notions of GAS in probability (GASp) and recurrence under only mild regularity conditions.
A stochastic version of the weak ISS result of \cite{sontag_further_1990} follows from \cite{teel_equivalent_2014}.
For (continuous-time) SDEs the literature is more sparse, especially due to the lack of a smooth converse Lyapunov function theorem.
While a converse theorem was proved in the 1960s in \cite{kushner1967converse}, the author of \cite{kushner1967converse} explicitely states in his survey \cite{kushner_partial_2014} that the Lyapunov function is not claimed to be smooth.
In \cite[Remark 5.5]{khasminskii_stochastic_2012} an example of an SDE is given where the origin is GASp, but no Lyapunov function differentiable at the origin exists.
A converse theorem that achieves a Lyapunov function $\mc C^2$ everywhere except maybe the origin is given in \cite[Theorem 5.12]{khasminskii_stochastic_2012} for stochastic exponential stability, but remains, to the best of the authors' knowledge, an open question for the weaker notion of GASp.

In this paper, we sidestep this gap by assuming the existence of a $\mc C^2$ (except maybe at the origin) Lyapunov function.
We then prove GASp for perturbed dynamics under a state-dependent perturbation bound vanishing at the origin, as well as a weak stochastic ISS property.
If a $\mc C^2$ (except maybe at the origin) converse theorem holds, then our assumption can be weakened to nominal GASp.
Note that if the nominal system is deterministic, then the converse theorem in \cite{lin1996smooth} becomes applicable, and it follows from our result that nominal GAS implies GASp under sufficiently small stochastic perturbations.
Furthermore, we prove that, under stronger assumptions, the allowed perturbation bound is at least proportional to the state near the origin.
Strengthening this even more, the setting of exponential stability is reached, where the aforementioned converse theorem \cite[Theorem 5.12]{khasminskii_stochastic_2012} applies.
Then we show that stochastic exponential stability is maintained under a globally proportional perturbation bound and implies a stochastic exponential ISS property.


Finally, we present a novel approach for backstepping of stochastic pure-feedback\footnote{That is, virtual control variables enter the next subsystem nonlinearly.} systems.
Stochastic backstepping in strict-feedback\footnote{That is, virtual control variables enter the next subsystem affinely.} form was first formulated in \cite{deng_stochastic_1997}.
The author of \cite{tsinias_stochastic_1998} constructs stabilizing controls for interconnections of systems that are not necessarily control-affine under stochastic ISS assumptions.
Adaptive neural control of stochastic pure-feedback systems was also done in, e.g., \cite{wang_iss-modular_2006, wang_robust_2013} under a uniform positivity assumption on control sensitivities, see \cite[Assumption 1]{wang_iss-modular_2006}.
In this paper, we aim to achieve integrator backstepping 
where the integrator state can enter the second subsystem non-affinely unlike \cite{deng_stochastic_1997}, and without the assumptions on ISS and control sensitivities of \cite{tsinias_stochastic_1998} and \cite{wang_iss-modular_2006,wang_robust_2013}, respectively.
For deterministic systems, the value of weak ISS for backstepping in pure-feedback form was shown in \cite{sontag1989remarks}.
There, under the aforementioned state-dependent scaling, a cascade of a GAS system and an ISS system is obtained, which is GAS by a special case of the ISS small gain theorem \cite{jiang_small-gain_1994}.
However, this fails for SDEs because the cascade of a GASp (or even GAS) and a stochastic ISS system is not necessarily GASp; see \cite[Proposition 6]{teel_converse_2014} for a counterexample in discrete time that can also be adapted to continuous time.
Instead, we make the same additional assumptions as for our stronger robustness result and construct a Lyapunov function for the cascade.
While this does not directly utilize our results on robustness and stochastic ISS, it still relies on a mean-value-theorem approach used for our stronger robustness result and commonly seen in backstepping.
Furthermore, we employ a state-dependent scaling in the Lyapunov function construction, which fulfills a role very similar to the scaling in weak ISS.

The rest of this paper is structured as follows.
In Section \ref{s:preliminaries} we recall the definitions of solutions as well as GASp and exponential $p$-stability for SDEs.
After stating the problem in Section \ref{s:problem_statement}, the main results on robustness and stochastic ISS are provided in Sections \ref{s:robustness} and \ref{s:weakISSp}, respectively.
Section \ref{s:backstepping} deals with backstepping.
Section \ref{s:proofs} contains some of the proofs, while Section \ref{s:conclusion} concludes the paper.

\textbf{Notation.}
We denote the closed and open unit balls of the considered space as $\overline\B$ and $\B$, respectively. The trace of a square matrix is denoted by $\text{Tr}$.
The notation $|\cdot|$ denotes the Euclidean norm for vectors and the spectral norm (largest singluar value) for matrices.
A function $\alpha:\R_{\geq0}\to\R_{\geq0}$ is said to be \emph{of class $\mc K_\infty$}, written as $\alpha\in\mc K_\infty$, if it is $0$ at $0$, continuous, strictly monotone increasing and unbounded.
A function $\varrho:\R^n\to\R_{\geq0}$ is \emph{positive definite}, written as $\varrho\in\mc{PD}$, if it is continuous, $0$ at $0$ and positive everywhere else.
We denote by $\mc C^2(\R^n)$ (resp. $\mc C_0^2(\R^n)$) the set of continuous functions $V:\R^n\to\R$ that are twice continuously differentiable on $\R^n$ (resp. $\R^n\setminus\{0\}$).
The expectation of a random variable is denoted by $\E[\cdot]$.


\newpage
\section{Preliminaries}\label{s:preliminaries}
\subsection{Strong solutions of SDEs}
In this section we recall from \cite{karatzas_brownian_1998} the notion of (strong) solutions for general time-varying SDEs
\begin{align}
    \td x = \overline f(t,x)\td t + \overline g(t,x)\td w,\label{eq:generalSDE}
\end{align}
where $x\in\R^n$ is the state and $\overline f:\R_{\geq0}\times\R^n\to\R^n$ and $\overline g:\R_{\geq0}\times\R^n\to\R^{n\times p}$ are measurable and locally Lipschitz in $x$ uniformly in $t$, that is, for any $R>0$ there exists $L_R\in\R_{\geq0}$ such that
    $|\overline f(t,x)-\overline f(t,y)|+|\overline g(t,x)-\overline g(t,y)| \leq L_R|x-y|$
for all $x,y\in R\overline\B$ and $t\geq0$.
Further suppose that $\overline f(t,0)=0$ and $\overline g(t,0)=0$ for all $t\geq0$, so that the origin is an equilibrium of \eqref{eq:generalSDE}.
Let $(\Omega,\mc F,\P)$ be a probability space with a filtration $\mc F_t, t\geq0$, and $w_t:\Omega\to\R^r, t\geq0$ an $r$-dimensional Brownian motion adapted to the filtration $\mc F_t, t\geq0$.


\begin{defn}\cite[Definition 5.2.1]{karatzas_brownian_1998}
    A \emph{(strong) solution} of \eqref{eq:generalSDE} with respect to the fixed Brownian motion $w_t$, initial state $\xi\in\R^n$ and initial time $t_0\geq0$ is a stochastic process $x_t:\Omega\to\R^n, t\geq t_0$
    with continuous sample paths (i.e., $t\mapsto x_t(\omega)$ is continuous for all $\omega\in\Omega$) such that the following conditions hold.
    \begin{enumerate}[label=(\roman*)]
        \item $x_t,t\geq t_0$ is adapted to the filtration $\mc F_t, t\geq t_0$,
        \item $\P[x_{t_0}=\xi]=1$,
        \item $\P[\int_{t_0}^t|\overline f_i(s,x_s)+\overline g_{ij}(s,x_s)|^2ds<\infty]=1$ holds for every $1\leq i\leq n$ and $1\leq j\leq r$ and $t\geq t_0$, and
        \item  the integral version
            $x_t = x_0 + \int_{t_0}^t\overline f(s,x_s)\td s + \int_{t_0}^t \overline g(s,x_s)\td w_s$ 
        of \eqref{eq:generalSDE} holds almost surely for any $t\geq t_0$.
    \end{enumerate}
        We say that a solution $x_t$ is \emph{unique} if $\P[x_t=\widetilde x_t~\forall t\geq t_0]=1$ for any solution $\widetilde x_t$.
        If a unique solution $x_t$ exists from an initial state $x_0$ and initial time $t_0$, we denote it by $\varphi(t;x_0,t_0)$.
\end{defn}
Note that we only consider complete solutions, that is, solutions defined for all $t\geq t_0$.
The local Lipschitz condition on $\overline f$ and $\overline g$ we made above is sufficient for uniqueness, as stated in \cite[Theorem 2.5]{karatzas_brownian_1998}.
According to \cite[Theorem 2.9]{karatzas_brownian_1998}, global Lipschitz continuity and a linear growth condition are sufficient for existence.

\subsection{Stability notions}
We now define GASp by following \cite[Sections 5.3, 5.4; Definition 5.1]{khasminskii_stochastic_2012}.

\begin{defn}\label{def:GASp}
    Suppose the SDE \eqref{eq:generalSDE} admits a solution for any initial state and initial time.
    Then the origin is said to be \emph{stable in probability} for \eqref{eq:generalSDE} if for any $\varepsilon>0$, $\eta>0$ and $t_0\geq0$ there exists $r>0$ such that for any initial state $x_0\in\R^n$ with $|x_0|<r$,
        $\P\left[\sup_{t\geq t_0}|\phi(t;x_0,t_0)|>\varepsilon\right] < \eta.$
    
    The origin is said to be \emph{globally asymptotically stable in probability (GASp)} for \eqref{eq:generalSDE} if it is stable in probability and, for all $x_0\in\R^n$ and $t_0\geq0$,
    $\P\left[\lim_{t\to\infty}\phi(t;x_0,t_0)=0\right]=1.$
\end{defn}

See \cite[Theorem 5.8]{khasminskii_stochastic_2012} for a Lyapunov characterization of GASp.
Furthermore, we consider the stronger notion of exponential $p$-stability, defined in \cite[Section 5.7]{khasminskii_stochastic_2012} as follows.
\begin{defn}\label{def:exponential_stability}
    Suppose the SDE \eqref{eq:generalSDE} admits a solution for any initial state and initial time.
    Then the origin is said to be \emph{exponentially $p$-stable} ($p>0$), if there exist $A,\alpha>0$ such that
    for any initial state $x_0\in\R^n$, initial time $t_0\geq0$ and time $t\geq t_0$,
        $\E\left[|\phi(t;x_0,t_0)|^p\right] \leq A|x_0|^p\exp(-\alpha (t-t_0))$.
\end{defn}
See \cite[Theorem 5.11]{khasminskii_stochastic_2012} for a Lyapunov theorem for exponential $p$-stability, and \cite[Theorem 5.12]{khasminskii_stochastic_2012} for a converse result.




\section{Problem statement} \label{s:problem_statement}

We now consider SDEs
\begin{align}
    \td x = f(x,u)\td t + g(x,u)\td w,\label{eq:SDE_system}
\end{align}
where $x\in\R^n$ is the state and $u\in\R^m$ represents a disturbance input or parametrizes model uncertainty.
Let $(\Omega,\mc F,\P)$, $\mc F_t$ and $w_t, t\geq0$, be as in Section \ref{s:preliminaries}.
Assume that $f:\R^n\times\R^m\to\R^n$ and $g:\R^n\times\R^m\to\R^{n\times p}$ are locally Lipschitz in their arguments.
Further assume that $f(0,0)=0$ and $g(0,0)=0$, so that the origin is an equilibrium when $u=0$.
The goal of this paper is to characterize stability of perturbed SDEs
\begin{align}
    \td x = f(x,\mf u(t,x))\td t + g(x, \mf u(t,x))\td w,\label{eq:perturbed_SDE}
\end{align}
where $\mf u(t,x)$ expresses the disturbance as function of time and state,
based on stability of the nominal (unperturbed) SDE
\begin{align}
    \td x = f(x,0)\td t + g(x,0)\td w \label{eq:SDE:zero_input}
\end{align}
when $\mf u(t,x)$ is sufficiently small.
In Section \ref{s:robustness} we give results where asymptotic stability is maintained for $\mf u(t,x)$ vanishing at $x=0$, while Section \ref{s:weakISSp} gives ISS guarantees that allow nonzero $\mf u(t,x)$ everywhere.
The majority of the proofs are given in Section \ref{s:proofs}.





In the remainder of this paper, we frequently use the differential operator $\mc L_u$ for system \eqref{eq:SDE_system}, defined as
\begin{align}
        \mc L_uV(x) &:= \tfrac{\del V(x)}{\del x}f(x,u)+ \tfrac12 \text{Tr}\left(g(x,u)^\top\tfrac{\del^2V(x)}{\del x^2}g(x,u)\right)
    \end{align}
for $V\in\mc C^2_0(\R^n)$, control input $u\in\R^m$ and state $x\in\R^n\setminus\{0\}$.

\section{Robustness of stochastic stability}\label{s:robustness}
We consider three cases with progressively stronger assumptions.
Section \ref{s:robustness_general} is the most general setting, where we show that existence of a $\mc C_0^2$ stochastic Lyapunov function (SLF) for \eqref{eq:SDE:zero_input} implies GASp of \eqref{eq:perturbed_SDE} if $|\mf u(t,x)|\leq\delta(x)$ for all $t\in\R_{\geq0}$ and $x\in\R^n$ for some sufficiently small $\delta\in\mc{PD}$.
While it would be preferable to directly assume GASp of \eqref{eq:SDE:zero_input} instead of its Lyapunov characterization, this would rely on a $\mc C_0^2$ converse theorem for GASp, which, to the best of the authors' knowledge, is an open question.
However, an important special case is when \eqref{eq:SDE:zero_input} is deterministic, i.e., $g(x,0)\equiv0$, in which case the converse theorem of \cite{lin1996smooth} can be applied and GAS(p) of $\td x = f(x,0)\td t$ implies GASp of \eqref{eq:perturbed_SDE} for $|\mf u(t,x)|\leq \delta(x)$ for all $t\in\R_{\geq0}$ and $x\in\R^n$, which we state in a corollary.

In Section \ref{s:stronger} we assert that, under stronger assumptions involving at least proportional Lyapunov decrease locally near the origin, the perturbation bound $\delta$ can also be proportional around the origin, that is, $\limsup_{x\to0}\delta(x)/|x|>0$.

Finally, in Section \ref{s:exponential} we strengthen the approach of Section \ref{s:stronger} from local to global via the notion of exponential $p$-stability as in Definition \ref{def:exponential_stability}.
Thanks to the converse theorem \cite[Theorem 5.12]{khasminskii_stochastic_2012}, we show that, under additional regularity of $f$ and $g$, exponential $p$-stability of \eqref{eq:SDE:zero_input} implies exponential $p$-stability of \eqref{eq:perturbed_SDE} under a proportional bound $|\mf u(t,x)|\leq d|x|$ for all $t\in\R_{\geq0}$ and $x\in\R^n$ for some sufficiently small $d>0$.

\subsection{Robustness of GASp (general case)}\label{s:robustness_general}
    
We assume the existence of a SLF for \eqref{eq:SDE:zero_input}.

\begin{assm}\label{assm:stability}
    Consider the SDE \eqref{eq:SDE:zero_input} and suppose there exist functions $V\in\mc C_0^2(\R^n),\alpha_1,\alpha_2\in\mc K_\infty$ and $\varrho\in\mc{PD}$ such that
    \begin{align}
        \alpha_1(|x|) \leq V(x) &\leq \alpha_2(|x|)\quad\forall x\in\R^n,\label{eq:assm:stability:1}\\
        \mc L_0V(x) &\leq -\varrho(x)~\quad\forall x\in\R^n\setminus\{0\}.\label{eq:assm:stability:2}
    \end{align}
\end{assm}
Note that $x=0$ is intentionally excluded in \eqref{eq:assm:stability:2} as well as the differentiability requirement of $V$, as \cite[Remark 5.5]{khasminskii_stochastic_2012} gives an example of an SDE where the origin is GASp, but no SLF twice continuously differentiable at the origin exists.


The next proposition extends the Lyapunov decrease \eqref{eq:assm:stability:2} to $\mc L_uV(x)\leq-\rho(x)$ for any $\rho\in\mc{PD}$ smaller than $\varrho$, as long as $u$ is small enough, expressed in the form $|u|\leq\delta(x)$ for some $\delta\in\mc{PD}$ depending on $\rho$.
Its proof is given in Section \ref{s:proof:prop:main}.

\begin{prop}\label{thm:main}
    Suppose Assumption \ref{assm:stability} holds.
    Then for any $\rho\in\mc{PD}$ with $\varrho-\rho\in\mc{PD}$
    there exists $\delta\in\mc{PD}$ such that
    \begin{align}
        \max_{|u|\leq\delta(x)}\mc L_uV(x) \leq -\rho(x)\quad\forall x\in\R^n\setminus\{0\}. \label{eq:lem:main}
    \end{align}
\end{prop}
We now state the main theorem, which guarantees robustness of GASp to perturbations $\mf u(t,x)$ with $|\mf u(t,x)|\leq\delta(x)$ for some $\delta\in\mc{PD}$, and is proved by combining Proposition \ref{thm:main} with the Lyapunov theorem \cite[Theorem 5.8]{khasminskii_stochastic_2012}.

\begin{thm}\label{thm:1.robustness}
    Suppose that Assumption \ref{assm:stability} holds. Then there exists $\delta\in\mc{PD}$ such that
    the origin is GASp for \eqref{eq:perturbed_SDE},
    where $\mf u:\R_{\geq0}\times\R^n\to\R^m$ is any function that is measurable, locally Lipschitz in $x$ uniformly in $t$, satisfies $|\mf u(t,x)|\leq\delta(x)$ for all $(t,x)\in\R_{\geq0}\times\R^n$,
    and for which \eqref{eq:perturbed_SDE} admits a unique solution at every initial state and initial time.
    
\end{thm}

\begin{proof}
    Let $\delta\in\mc{PD}$ be as in Proposition \ref{thm:main}, when $\rho$ is chosen as $\rho(x):=\varrho(x)/2$.
    By Proposition \ref{thm:main}, 
        $\mc L_{\mf u(t,x)}V(x) \leq \max_{|u|\leq\delta(x)}\mc L_uV(x) \leq -\varrho(x)/2$
    for all $(t,x)\in\R_{\geq0}\times(\R^n\setminus\{0\})$.
    The Lyapunov theorem \cite[Theorem 5.8]{khasminskii_stochastic_2012} then implies that the origin is GASp for \eqref{eq:perturbed_SDE}.
\end{proof}

An important special case is when \eqref{eq:SDE:zero_input} is deterministic, i.e., $g(x,0)=0$ for all $x\in\R^n$.
Then GASp as in Definition \ref{def:GASp} of $\td x = f(x,0)\td t$ is equivalent to the standard (uniform) GAS notion \cite[Definition 2.2]{lin1996smooth} of the ordinary differential equation $\dot x = f(x,0)$; note that attractivity implies uniform attractivity for autonomous systems \cite[Proposition 2.14]{bacciotti_liapunov_2005}.
Then the smooth converse theorem \cite[Theorem 2.8]{lin1996smooth} can be applied to obtain the following corollary of Theorem \ref{thm:1.robustness}.

\begin{corollary}\label{cor:deterministic_nominal}
    Consider the SDE \eqref{eq:SDE_system} and suppose $g(x,0)=0$ for all $x\in\R^n$ and that the origin is GAS(p) for \eqref{eq:SDE:zero_input}.
    Then Assumption \ref{assm:stability} and the statement of Theorem \ref{thm:1.robustness} hold.
\end{corollary}
In particular, Corollary \ref{cor:deterministic_nominal} includes systems of the form $\td x = f(x)\td t + u\,\td w$, where $u$ is a ``virtual input" that is artificially introduced to check if stability of $\dot x=f(x)$ implies stochastic stability with noise.

\subsection{Locally proportional perturbation bound under additional assumptions}\label{s:stronger}
We now
show a strengthened version of Theorem \ref{thm:1.robustness} where $\delta(x)$ is at least proportional to $x$ locally around the origin under the following assumption.

\begin{assm}\label{assm:quadratic_rho}
    The functions $f$ and $g$ are continuously differentiable and Assumption \ref{assm:stability} holds such that 
    \begin{align}
        \limsup_{x\to0}\left|\tfrac{\del^i V(x)}{\del x^i}\right|\tfrac{|x|^i}{\varrho(x)}&<\infty,\quad i=0,1,2\label{eq:assm:stronger:2}
    \end{align}
    under the convention $\frac{\del^0V(x)}{\del x^0}=V(x)$.
\end{assm}
Condition \eqref{eq:assm:stronger:2} asserts locally proportional Lyapunov decrease for $i=0$, along with corresponding bounds on the gradient and Hessian of $V$ for $i=1,2$.
\begin{remark}\label{rem:quadratic}
    If $V\in\mc C^2(\R^n)$, i.e., $\mc C^2$ even at 0,
    then
    \begin{align}
        \limsup_{x\to0}\left|\tfrac{\del^iV(x)}{\del x^i}\right||x|^{i-2}<\infty,\quad i=0,1,2.
        \label{eq:remark_quadratic}
    \end{align}
    For $i=2$ this is because the Hessian of $V\in\mc C^2(\R^n)$ is continuous and thereby locally bounded. For $i=1$, \eqref{eq:remark_quadratic} follows from the fundamental theorem of calculus, as $\tfrac{\del V(x)}{\del x} = x^\top V_x(x)$ for all $x\in\R^n$
    with continuous $V_x:\R^n\to\R^{n\times n}$ defined as $V_x(x) = \int_0^1\tfrac{\del^2 V(sx)}{\del x^2}ds$, and similar for $i=0$. 
    By \eqref{eq:remark_quadratic}, the condition $\liminf_{x\to0}\tfrac{\varrho(x)}{|x|^2}>0$ (i.e., locally at least quadratic Lyapunov decrease) is then sufficient for \eqref{eq:assm:stronger:2}.
\end{remark}
Under Assumption \ref{assm:quadratic_rho}, Theorem \ref{thm:1.robustness} is strengthened as follows.

\begin{thm}\label{thm:stronger}
    Suppose Assumption \ref{assm:quadratic_rho} holds.
    Then the statement of Theorem \ref{thm:1.robustness} holds, where $\delta$ can be chosen to satisfy $\liminf_{x\to0}\delta(x)/|x|>0$.
\end{thm}
In the following, we derive the main idea to prove Theorem~\ref{thm:stronger}, which is also the approach for stochastic backstepping in Section \ref{s:backstepping}.
Under Assumption \ref{assm:quadratic_rho}, 
$f$ and $g$ can be written~as
\begin{align}
    f(x,u) &= f(x,0) + f_u(x,u)u,\label{eq:stronger:fu}\\
    g(x,u) &= g(x,0) + \begin{bmatrix}g_u^1(x,u)u&\dots&g_u^p(x,u)u\end{bmatrix}\negthickspace,
    \label{eq:stronger:gu}\\
    g(x,0) &=  \begin{bmatrix}g_x^1(x)x&\dots&g_x^p(x)x\end{bmatrix}\negthickspace,
    \label{eq:stronger:gu0}
\end{align}
for all $(x,u)\in\R^n\times\R^m$ with continuous functions $f_u, g_u^i:\R^n\times\R^m\to\R^{n\times m}$ and $g_x^i:\R^n\to\R^{n\times m}$, where $i\in\{1,\dots,p\}$.
This is commonly used in backstepping, see, e.g., \cite{krstic1995nonlinear,wang_iss-modular_2006}, and also see Hadamard's lemma.
The function $f_u$ can be constructed as $f_u(x,u):=\int_0^1 \tfrac{\del f(x,su)}{\del u}ds$, where \eqref{eq:stronger:fu} follows from the fundamental theorem of calculus.
$g_u^i(x,u)$ and $g_x^i(x)$ are constructed similarly for each column of $g(x,u)$ and $g(x,0)$.
For brevity, we write
    $g(x,u) = g(x,0) + g_u(x,u)\otimes u$
and $g(x,0) = g_x(x)\otimes x$, where $\otimes$ denotes the tensor product.
With that, we have that for $x\in\R^n\setminus\{0\}$ and $u\in\R^m$,
\begin{align*}
    \mc L_uV(x)
    &= \tfrac{\del V(x)}{\del x}f(x,u)\nonumber
        + \tfrac12 \text{Tr}\left(g(x,u)^\top\tfrac{\del^2V(x)}{\del x^2}g(x,u)\right)\\
    &= \tfrac{\del V(x)}{\del x}(f(x,0) + f_u(x,u)u)\nonumber\\
    &~~~~+ \tfrac12 \text{Tr}\Big(\left(g(x,0) + g_u(x,u)\otimes u\right)^\top\tfrac{\del^2V(x)}{\del x^2}\\
    &~~~~~~~~~~~~~~~~~~~~~~~~~~~~~~\left(g(x,0) + g_u(x,u)\otimes u\right)\Big)\\
    &= \mc L_0V(x) + \tfrac{\del V(x)}{\del x}f_u(x,u)u\\
    &~~~~+\text{Tr}\left( (g_x(x)\otimes x)^\top \tfrac{\del^2 V(x)}{\del x^2}(g_u(x,u)\otimes u)\right)\\
    &~~~~+\tfrac12 \text{Tr}\left( (g_u(x,u)\otimes u)^\top \tfrac{\del^2 V(x)}{\del x^2} (g_u(x,u)\otimes u)\right)\\
    &= \mc L_0V(x) + \tfrac{\del V(x)}{\del x}f_u(x,u)u\\
    &~~~~+\sum_{i=1}^p x^\top g_x^i(x)^\top \tfrac{\del^2 V(x)}{\del x^2}g_u^i(x,u)u\\
    &~~~~+\sum_{i=1}^p u^\top g_u^i(x,u)^\top\tfrac{\del^2 V(x)}{\del x^2}g_u^i(x,u)u\numberthis\label{eq:LuV_stronger}\\
    &= \mc L_0 V(x) + O\left(\varrho(x)\left(\tfrac{|u|}{|x|}+\tfrac{|u|^2}{|x|^2}\right)\right) \label{eq:LuV_stronger_final}\numberthis
\end{align*}
in Landau notation as $(x,u)\to0$ by \eqref{eq:assm:stronger:2}.
This, along with \eqref{eq:assm:stability:2}, enables the locally proportional perturbation bound.
See Section \ref{s:proof:thm:stronger} for the details of the proof of Theorem \ref{thm:stronger}.


\subsection{Robust exponential stability with globally proportional perturbation bound}\label{s:exponential}
We further strengthen the methods of Section \ref{s:stronger} to prove robustness of $p$-exponential stability for $p>1$ with a globally proportional perturbation bound.
Thanks to the converse theorem \cite[Theorem 5.12]{khasminskii_stochastic_2012} we can directly assume exponential $p$-stability, with regularity of $f$ and $g$ as follows.
\begin{assm}\label{assm:exp_stability}
    The origin of \eqref{eq:SDE:zero_input} is exponentially $p$-stable ($p>1$) and $f$ and $g$ have continuous bounded first and second derivatives.
\end{assm}
Assumption \ref{assm:exp_stability} implies the existence of a SLF with a bound for $\mc L_uV(x)$ similar to \eqref{eq:LuV_stronger_final} with $\varrho(x) = |x|^p$ as follows.
\begin{prop}\label{lem:exponential}
    Suppose that Assumption \ref{assm:exp_stability} holds\footnote{We require $p>1$ to prove \eqref{eq:def:expSLF_4}, whereas \eqref{eq:def:expSLF_1}-\eqref{eq:def:expSLF_3} can be achieved with any $p>0$, see \cite[Theorem 5.12]{khasminskii_stochastic_2012}.}.
    Then there exist $V\in\mc C_0^2(\R^n)$ and $a_1,a_2,a_3,a_4,C>0$ such that
    \begin{align}
        a_1|x|^p \leq V(x) &\leq a_2|x|^p, 
        \label{eq:def:expSLF_1}\\
        \mc L_0V(x) &\leq -a_3|x|^p,
        \label{eq:def:expSLF_2}\\
        \left|\tfrac{\del^iV(x)}{\del x^i}\right| &\leq a_4|x|^{p-i},\quad i=1,2,\label{eq:def:expSLF_3}\\
        \mc L_uV(x) &\leq \mc L_0V(x) + C|x|^{p-2}(|x||u| + |u|^2)\label{eq:def:expSLF_4}
    \end{align}
    hold for all $(x,u)\in(\R^n\setminus\{0\})\times\R^m$.
\end{prop}
See Section \ref{proof:prop:exponential} for the proof. Note that, by \eqref{eq:def:expSLF_2} and \eqref{eq:def:expSLF_3}, Assumption \ref{assm:exp_stability} implies Assumption \ref{assm:quadratic_rho}.
Remark \ref{rem:quadratic} corresponds to the case where $p=2$.
Combining Proposition \ref{lem:exponential} with the Lyapunov theorem \cite[Theorem 5.11]{khasminskii_stochastic_2012} gives the following stronger version of Theorems \ref{thm:1.robustness} and \ref{thm:stronger}.

\begin{thm}
    Suppose Assumption \ref{assm:exp_stability} holds.
    Then the statement of Theorem \ref{thm:1.robustness} holds with exponential $p$-stability instead of GASp and $\delta(s)=ds$ for some $d>0$.
\end{thm}
\begin{proof}
    We apply Proposition \ref{lem:exponential} and choose $d:=\min\left\{\tfrac{a_3}{4C},\sqrt{\tfrac{a_3}{4C}}\right\}$. Then, by \eqref{eq:def:expSLF_2} and \eqref{eq:def:expSLF_4},
        $\mc L_uV(x) \leq -a_3|x|^p + C(d+d^2)|x|^p \leq -0.5a_3|x|^p$
    for all $x\in\R^n\setminus\{0\}$ and $u\in\R^m$ with $|u|\leq d|x|$.
    Exponential $p$-stability of \eqref{eq:perturbed_SDE} for $\mf u$ as in Theorem \ref{thm:1.robustness} with $\delta(s)=ds$ then follows from \cite[Theorem 5.11]{khasminskii_stochastic_2012}. 
\end{proof}



\section{Input-to-state-stability}\label{s:weakISSp}
With the robustness results for perturbations vanishing at the origin established, we now consider input-to-state stability, which allows non-zero perturbations at origin.
The following two subsections correspond to the general setting of Section \ref{s:robustness_general} and the stronger exponential setting of Section \ref{s:exponential}, respectively.


\subsection{Weak input-to-state stability}\label{ss:weakISSp}
The goal of this section is to derive a stochastic version of \cite[Theorem 1]{sontag_further_1990} which states that, for deterministic systems, stability implies \emph{weak} input-to-state stability, with \emph{weak} meaning under a state-dependent positive input scaling.
We state, under Assumption \ref{assm:stability}, stochastic ISS (SISS) holds under an input transformation $u=\gamma(x)v$ with respect to the new input $v$, where $\gamma:\R^n\to\R_{>0}$ is a suitable scaling function.
Function $\gamma$ plays a similar role to $\delta$ in the previous section, as SISS with respect to $u$ may not hold due to nonlinearity in the input.
However, $\gamma$ can be positive at $0$, whereas $\delta$ in Theorem \ref{thm:1.robustness} usually has to be $0$ at $0$.


We start by defining input-to-state stability in probability, followed by a sufficient Lyapunov condition and eventually the main theorem of this section.

\begin{defn}\cite[Definition 3]{liu_notion_2008}\label{def:SISS}
    System \eqref{eq:SDE_system} is \emph{stochastically input-to-state stable (SISS)} if for any $\varepsilon>0$ there exist $\beta\in\mc{KL}$ and $\zeta\in\mc K$ such that
        $\P[|x_t|\leq\beta(|x_0|,t) + \zeta(\text{\normalfont ess\,sup}_{0\leq s\leq t}|\mf u(x_s,s)|)] \geq 1-\varepsilon,$
    where $\text{\normalfont ess\,sup}$ denotes the essential supremum,
    holds for all $t\geq0$ for the solution $x_t$ of \eqref{eq:perturbed_SDE}
    from any initial state $x_0$ and where $\mf u:\R_{\geq0}\times\R^n\to\R^m$ is any function that is measurable, essentially bounded, locally Lipschitz in $x$ uniformly in $t$, and for which \eqref{eq:perturbed_SDE} admits a unique solution from any initial state.
\end{defn}

\begin{defn}\label{defn:SISS-LF}
    A function $V\in\mc C^2_0(\R^n)$ is said to be a \emph{SISS-Lyapunov function (SISS-LF)} for system \eqref{eq:SDE_system} with input $u$ if there exist $\alpha_1,\alpha_2,\kappa\in\mc K_\infty$ and $\varrho\in\mc{PD}$ such that
    \begin{align}
        \alpha_1(|x|) \leq V(x) \leq \alpha_2(|x|)\quad&\forall x\in\R^n,\\
        \max_{|u|\leq\kappa(|x|)}\mc L_uV(x) \leq -\varrho(x)~\quad&\forall x\in\R^n\setminus\{0\}. \label{eq:defn:SISS_2}
    \end{align}
\end{defn}

Definition \ref{defn:SISS-LF} is inspired by \cite{teel_equivalent_2014} and resembles \eqref{eq:lem:main}, where the difference here is that $\kappa$ needs to be in $\mc K_\infty$.
Definition \ref{defn:SISS-LF} is less strict than the definition in \cite{liu_notion_2008} which requires $\mc L_uV(x) \leq \chi(|u|) - \alpha(|x|)$ for some $\chi,\alpha\in\mc K_\infty$.
Every SISS-LF in the sense of \cite{liu_notion_2008} satisfies Definition \ref{defn:SISS-LF} (choose $\varrho(x):=\alpha(|x|)/2$ and $\kappa(s) := \chi^{-1}(\alpha(s)/2)$), but not conversely as $\varrho$ is generally not radially unbounded.
Nevertheless, a SISS-LFs according to Definition \ref{defn:SISS-LF} is enough to guarantee the same SISS notion as in \cite{liu_notion_2008} (and our Definition \ref{def:SISS}) by following the proof of \cite[Theorem 1]{liu_notion_2008}.
We are now ready to state the main theorem of this section.
\begin{thm}\label{thm:main2}
    Suppose Assumption \ref{assm:stability} holds. Then there exists a continuous function $\gamma:\R^n\to\R_{>0}$ such that $V$ is a SISS-LF for
    \begin{align}
        \td x = f(x,\gamma(x)v) + g(x,\gamma(x)v)\td w.\label{eq:inflated_SDE}
    \end{align}
    with input $v$.
    In particular, \eqref{eq:inflated_SDE} is SISS with input $v$.
\end{thm}
\begin{proof}
    Let $\delta\in\mc{PD}$ be as in Proposition \ref{thm:main} for $\rho:=\varrho/2$.
    Since $\delta\in\mc{PD}$, by \cite[Lemma 18]{kellett_compendium_2014} there exist continuous $\widetilde\gamma:\R_{\geq0}\to\R_{>0}$ and $\kappa\in\mc K_\infty$ such that $\widetilde\gamma(|x|)\kappa(|x|) \leq \delta(x)$ for all $x\in\R^n$.
    Define $\gamma:\R^n\to\R_{>0}, x\mapsto\widetilde\gamma(|x|)$, then $\gamma(x)\kappa(|x|)\leq\delta(x)$ for all $x\in\R^n$.
    In particular, we have for all $v\in\R^m$ with $|v|\leq\kappa(|x|)$, that $|\gamma(x)v|\leq\delta(x)$.
    With this, \eqref{eq:lem:main} implies
        $\max_{|v|\leq\kappa(|x|)}\mc L_{\gamma(x)v}V(x) \leq -\varrho(x)/2$
    for any $x\in\R^n\setminus\{0\}$.
    This, along with \eqref{eq:assm:stability:1}, verifies that $V$ is an SISS-LF for \eqref{eq:inflated_SDE} with control input $v$.
    It then follows from \cite[Theorem 1]{liu_notion_2008} that \eqref{eq:inflated_SDE} is SISS with respect to the input $v$.
\end{proof}

By replacing $f(x,u)$ with $f(x,k(x)+u)$ for some feedback control law $k:\R^n\to\R^m$, a result  similar to \cite[Theorem~1]{sontag_further_1990} on weak input-to-state stabilizability can be obtained.

\subsection{Exponential input-to-state stability}\label{s:exponential_ISS}
We now revisit the $p$-exponential framework of Section \ref{s:exponential} for $p\geq2$.
Then the state-dependent scaling $\gamma(x)$ of Section \ref{ss:weakISSp} is not required and the following form of exponential ISS holds when the nominal system is exponentially $p$-stable.

\begin{thm}
    Suppose Assumption \ref{assm:exp_stability} holds for $p\geq2$.
    Then there exists $\alpha>0$ and $A,B>0$ as well as $A_\varepsilon,B_\varepsilon>0$ for any $\varepsilon>0$, such that
    \begin{align}
        \E&\left[|x_t|^p\right] \leq A|x_0|^pe^{-\alpha t} + B{\overline u}^p,\label{eq:thm:expISS:1}\\
        \P&\left[|x_t| \leq A_\varepsilon|x_0|e^{-\alpha t/p} + B_\varepsilon\overline u\right] \geq 1-\varepsilon\label{eq:thm:expISS:2}
    \end{align}
    for any $t\geq0$ and any solution $x_t$ to \eqref{eq:perturbed_SDE} as in Definition \ref{def:SISS}, where $\overline u := \text{\normalfont ess\,sup}_{0\leq s\leq t}|\mf u(x_s,s)|$.
\end{thm}
\begin{proof}    
    Let $V,a_1,a_2,a_3,a_4$ and $C$ be as in Proposition~\ref{lem:exponential}.
    Using weighted Young's inequality, there exists a constant $K>0$ depending on $a_3,C$ and $p$ such that
    \begin{align}
        -a_3|x|^p + C|x|^{p-2}(|x||u|+|u|^2)\leq -\tfrac{a_3}2|x|^p + K|u|^p\label{eq:proof:expISS}
    \end{align}
    holds for all $(x,u)\in\R^n\times\R^m$.
    Now let $x_t,t\geq0$ be a solution to \eqref{eq:perturbed_SDE} as in Definition \ref{def:SISS} and $\overline u := \text{\normalfont ess\,sup}_{0\leq s\leq t}|\mf u(x_s,s)|$.
    Using \eqref{eq:proof:expISS} along with Itô's lemma, \eqref{eq:def:expSLF_4}, \eqref{eq:def:expSLF_2} and \eqref{eq:def:expSLF_1}, for almost every $t\geq0$,
        $\tfrac{\td}{\td t}\E[V(x_t)]
        = \E\left[\mc L_{\mf u(t, x_t)}V(x_t)\right]\nonumber
        \leq \E\left[-\tfrac{a_3}2|x_t|^p + K|\mf u(t, x_t)|^p\right]\nonumber
        \leq-\tfrac{a_3}{2a_2}\left(\E[V(x_t)]-\tfrac{2Ka_2}{a_3}\overline u^p\right).$
    Then, by Gr\"onwall's inequality,
        $\E[V(x_t)] - \tfrac{2Ka_2}{a_3}\overline u^p \leq \left(V(x_0) - \tfrac{2Ka_2}{a_3}\overline u^p\right) \exp\left(-\tfrac{a_3}{2a_2}t\right)\nonumber
        \leq V(x_0) \exp\left(-\tfrac{a_3}{2a_2}t\right)$
    for any $t\geq0$.
    Using \eqref{eq:def:expSLF_1}, we then get for any $t\geq0$ that
    \begin{align}
        \E\left[|x_t|^p\right] \leq \underbrace{\tfrac{a_2}{a_1}}_{=:A}|x_0|^p\exp\Big(-\underbrace{\tfrac{a_3}{2a_2}}_{=:\alpha}t\Big) + \underbrace{\tfrac{2Ka_2}{a_1a_3}}_{=:B}\overline u^p,
    \end{align}
    which shows \eqref{eq:thm:expISS:1}.
    Inequality \eqref{eq:thm:expISS:2} follows for any $\varepsilon>0$ with $A_\varepsilon:=\left(\tfrac A{\varepsilon}\right)^{1/p}, B_\varepsilon:=\left(\tfrac B{\varepsilon}\right)^{1/p}$ from \eqref{eq:thm:expISS:1}, Markov's inequality and $a^p+b^p\leq(a+b)^p$ for any $a,b\geq0,p\geq1$.
\end{proof}

\section{Application to stochastic backstepping}\label{s:backstepping}
\subsection{Class of systems}
Consider now the cascaded system
\begin{subequations}
\begin{align}
    \td x &= f(x,u)\td t + g(x,u)\td w, \label{eq:SBS:1}\\
    \dot u &= v,\label{eq:SBS:2}
\end{align}
\end{subequations}
with $f$ and $g$ as in Section \ref{s:robustness} and scalar $u$ and $v$.
Our goal is to construct a control law $v=v(x,u)$ such that the cascade is GASp.
We make the next assumption 
on system \eqref{eq:SBS:1} that there exists a static state-feedback law such that the corresponding closed-loop system satisfies Assumption \ref{assm:stability} with the additional conditions of Remark \ref{rem:quadratic} and further regularity.
\begin{assm}\label{assm:backstepping}
    There exists a twice continuously differentiable feedback control law $k:\R^n\to\R$ such that the conditions of Assumption \ref{assm:stability} hold for system \eqref{eq:SBS:1} under $u=k(x)$,
    with $V\in\mc C^2(\R^n)$ and $\liminf_{x\to0}\varrho(x)/|x|^2>0$.
    Furthermore, $f$ and $g$ have locally Lipschitz derivatives, and $V$ and $k$ have locally Lipschitz Hessians.
\end{assm}

\begin{remark}
    Our procedure works similarly for the more general dynamics $du = (v + f_2(x,u))\td t + g_2(x,u)\td w$ with appropriate $f_2, g_2$, but we consider \eqref{eq:SBS:2} for simplicity.
\end{remark}

\subsection{Control design in the general case}\label{ss:backstepping:design}
We define $z:=u-k(x)$ and 
\begin{subequations}
\begin{align}
    f_k(x,z)&:=f(x,k(x)+z),\label{eq:defn:fk}\\
    g_k(x,z)&:=g(x,k(x)+z).\label{eq:defn:gk}
\end{align}
\end{subequations}
In transformed coordinates $(x,z)$ the dynamics \eqref{eq:SBS:1} and \eqref{eq:SBS:2} become, by It\^o's lemma,
\begin{subequations}
\begin{align}
    \td x &= f_k(x,z)\td t + g_k(x,z)\td w, \label{eq:SBS:1n}\\
    dz &= (v - \mc L_{k(x)+z}k(x))\td t - \tfrac{\del k(x)}{\del x}g_k(x,z)\td w,\label{eq:SBS:2n}
\end{align}
\end{subequations}
which we write as
\begin{subequations}
\begin{align}
    \td\widetilde x = \widetilde f(&\widetilde x)\td t + \widetilde g(\widetilde x)\td w, \label{eq:backstepping:system_tilde}\\
    \widetilde x := \begin{bmatrix}x\\z \end{bmatrix},\quad \widetilde f(\widetilde x) &:= \begin{bmatrix}
        f_k(x,z)\\v - \mc L_{k(x)+z}k(x)
    \end{bmatrix},\\
    \widetilde g(\widetilde x) &:= \begin{bmatrix}
        g_k(x,z)\\- \tfrac{\del k(x)}{\del x}g_k(x,z)
    \end{bmatrix}.
\end{align}
\end{subequations}
Our target SLF is
\begin{align}
    \widetilde V(\widetilde x):=V(x)+\tfrac12\gamma(x) z^2,
\end{align}
with a $\mc C^\infty$ function $\gamma:\R^n\to\R_{>0}$ yet to be defined.
We aim to construct $v=v(x,u)$ such that
\begin{align}
    \tmc L\widetilde V(\widetilde x) \leq -\varrho(x)/2 - cz^2\quad\forall (x,z)\in\R^n\times\R\label{eq:backstepping:Lapunov_decrease}
\end{align}
for some $c>0$, where $\tmc L$ denotes the differential operator of the system \eqref{eq:backstepping:system_tilde}.
By definition of $\widetilde V$, the first derivative and Hessian of $\wt V$ are
\begin{align}
    \tfrac{\del \widetilde V(\wt x)}{\del\wt x} &= \begin{bmatrix}\tfrac{\del V(x)}{\del x} + \tfrac12\tfrac{\del\gamma(x)}{\del x}z^2 & \gamma(x)z\end{bmatrix},\\
    \tfrac{\del^2 \wt V(\wt x)}{\del\wt x^2} &= \begin{bmatrix}
        \tfrac{\del^2 V(x)}{\del x^2} + \tfrac12\tfrac{\del^2\gamma(x)}{\del x^2}z^2 & \tfrac{\del\gamma(x)}{\del x}^\top z \\ \tfrac{\del\gamma(x)}{\del x} z & \gamma(x)
    \end{bmatrix}.
\end{align}
Similar to Section \ref{s:stronger}, we write
\begin{align}
    f_k(x,z) &= f_k(x,0) + f_z(x,z)z,\\
    g_k(x,z) &= g_k(x,0) + g_z(x,z) z\nonumber\\
    &= g_k(x,0) + \begin{bmatrix}g_z^1(x,z)z&\dots&g_z^p(x,z)z\end{bmatrix},\\
    g_k(x,0) &= \begin{bmatrix}g_x^1(x)x&\dots&g_x^p(x)x\end{bmatrix}\label{eq:backstepping_gx}
\end{align}
for suitable $f_z,g_z,g_z^i$ and $g_x^i$.
Furthermore, let $V_x$ be as in Remark \ref{rem:quadratic}.
Then, as in \eqref{eq:LuV_stronger},
\begin{align}
    \mc L_{k(x)+z}V(x)\!-\!\mc L_{k(x)}V(x)&=x^\top\eta_1(x,z)z\!+\!\eta_2(x,z)z^2,
\end{align}
with $\eta_1,\eta_2:\R^n\times\R\to\R^n$ defined as
\begin{align}
    \eta_1(x,z) &:= V_x(x)f_z(x,z) + \sum_{i=1}^p g_x^i(x)^\top\tfrac{\del^2 V(x)}{\del x^2}g_z^i(x,z),\nonumber\\[-5pt]
    \eta_2(x,z) &:= \sum_{i=1}^p g_z^i(x,z)^\top\tfrac{\del^2 V(x)}{\del x^2}g_z^i(x,z).
\end{align}
Then, for all $\wt x\in\R^n\times\R$,
\begin{align*}
    &\tmc L\wt V(\wt x)\\[-0.5pt]
    &= \tfrac{\del \widetilde V(\wt x)}{\del\wt x}\wt f(\wt x) + \tfrac12\text{Tr}\left(\wt g(\wt x)^\top \tfrac{\del^2 \wt V(\wt x)}{\del\wt x^2} \wt g(\wt x)\right)\\[-0.5pt]
    &= \tfrac{\del V(x)}{\del x}f_k(x,z) + \tfrac12\tfrac{\del\gamma(x)}{\del x}z^2 f_k(x,z)\\[-0.5pt]  
    &~~~+ \gamma(x)z(v-\mc L_{k(x)+z}k(x))\\[-0.5pt]
    &~~~+\tfrac12 \text{Tr}\left(g_k(x,z)^\top\left(\tfrac{\del^2 V(x)}{\del x^2}+\tfrac12\tfrac{\del^2\gamma(x)}{\del x^2}z^2\right)g_k(x,z)\right)\\[-0.5pt]
    &~~~-\text{Tr}\left(g_k(x,z)^\top \tfrac{\del\gamma(x)}{\del x}^\top z\tfrac{\del k(x)}{\del x}g_k(x,z)\right)\\[-0.5pt]
    &~~~+\tfrac12\text{Tr}\left(g_k(x,z)^\top \tfrac{\del k(x)}{\del x}^\top\gamma(x)\tfrac{\del k(x)}{\del x}g_k(x,z)\right)\\[-0.5pt]
    &= \mc L_{k(x)+z}V(x)\\[-0.5pt]
    &~~~+ z\Big[\gamma(x)(v-\mc L_{k(x)+z}k(x)) + \tfrac12\tfrac{\del\gamma(x)}{\del x}z f_k(x,z)\\[-0.5pt]
    &~~~~~~~~~~~~+\tfrac14 \text{Tr}\left(g_k(x,z)^\top\tfrac{\del^2\gamma(x)}{\del x^2}z g_k(x,z)\right)\\[-0.5pt]
    &~~~~~~~~~~~~-\text{Tr}\left(g_k(x,z)^\top \tfrac{\del\gamma(x)}{\del x}^\top\tfrac{\del k(x)}{\del x}g_k(x,z)\right)\Big]\\[-0.5pt] 
    &~~~+\tfrac12\text{Tr}\Big(\left(g_k(x,0)+g_z(x,z)z\right)^\top \tfrac{\del k(x)}{\del x}^\top\gamma(x)\\[-0.5pt]
    &~~~~~~~~~~~~~~~~~~~~~~~~~~~~~~~\tfrac{\del k(x)}{\del x}(g_k(x,0)+g_z(x,z)z)\Big)\\[-0.5pt]
    &= \mc L_{k(x)}V(x) + z[x^\top\eta_1(x,z) + \eta_2(x,z)z]\\[-0.5pt]
    &~~~+ z\Big[\gamma(x)(v-\mc L_{k(x)+z}k(x)) + \tfrac12\tfrac{\del\gamma(x)}{\del x}z f_k(x,z)\\[-0.5pt]
    &~~~~~~~~~~~~+\tfrac14 \text{Tr}\left(g_k(x,z)^\top\tfrac{\del^2\gamma(x)}{\del x^2}z g_k(x,z)\right)\\[-0.5pt]
    &~~~~~~~~~~~~-\text{Tr}\left(g_k(x,z)^\top \tfrac{\del\gamma(x)}{\del x}^\top\tfrac{\del k(x)}{\del x}g_k(x,z)\right)\\[-0.5pt]
    &~~~~~~~~~~~~+\text{Tr}\left(g_k(x,0)^\top \tfrac{\del k(x)}{\del x}^\top\gamma(x)\tfrac{\del k(x)}{\del x}g_z(x,z)\right)\\[-0.5pt]
    &~~~~~~~~~~~~+\tfrac12\text{Tr}\left(g_z(x,z)^\top \tfrac{\del k(x)}{\del x}^\top\gamma(x)\tfrac{\del k(x)}{\del x}g_z(x,z)z\right)\Big]\\[-0.5pt] 
    &~~~+\tfrac12\text{Tr}\left(g_k(x,0)^\top \tfrac{\del k(x)}{\del x}^\top\gamma(x)\tfrac{\del k(x)}{\del x}g_k(x,0)\right). \numberthis\label{eq:backstepping:calculation}
\end{align*}
All terms in square brackets can be compensated by suitable choice of $v$.
The final term, which only depends on $x$, can be compensated by $\varrho(x)/2$ and suitable choice of $\gamma$ as follows.
By \eqref{eq:backstepping_gx} and the definition of the spectral norm $|\cdot|$,
\begin{align}
    &\tfrac12\text{Tr}\left(g_k(x,0)^\top \tfrac{\del k(x)}{\del x}^\top\gamma(x)\tfrac{\del k(x)}{\del x}g_k(x,0)\right)\nonumber\\
    &= \tfrac12\sum_{i=1}^p x^\top g_x^i(x)^\top \tfrac{\del k(x)}{\del x}^\top \gamma(x)\tfrac{\del k(x)}{\del x}g_x^i(x)x\nonumber\\
    &\leq \tfrac12\gamma(x)\underbrace{\n{\sum_{i=1}^p g_x^i(x)^\top \tfrac{\del k(x)}{\del x}^\top\tfrac{\del k(x)}{\del x}g_x^i(x)}}_{=:\mu(x)}|x|^2\label{eq:defn:mu}
\end{align}
for every $x\in\R^n$.
The function $\mu:\R^n\to\R_{\geq0}$ is continuous since $g_k^i$ and $\tfrac{\del k}{\del x}$ are continuous.
We define $\overline\gamma:\R^n\to\R_{>0}$ as
\begin{align}    \overline\gamma(0)&:=\begin{cases}\liminf_{x\to0}\tfrac{\varrho(x)}{\mu(x)|x|^2}, & \mu(0)>0,\\
    1, & \mu(0)=0,\end{cases}\label{eq:def:gamma0}\\
    \overline\gamma(x) &:= \min\left\{\overline\gamma(0),\tfrac{\varrho(x)}{\mu(x)|x|^2}\right\}\quad\forall x\in\R^n\setminus\{0\}.\label{eq:def:gammax}
\end{align}
Since $\liminf_{x\to0}\varrho(x)/|x|^2>0$ by Assumption \ref{assm:backstepping} and since $\mu$ is continuous at $0$, we have $\overline\gamma(0)>0$.
Furthermore, because $\varrho\in\mc{PD}$, we have $\overline\gamma(x)>0$ for all $x\in\R^n$.
The function $\overline\gamma$ is continuous at every $x\in\R^n\setminus\{0\}$ with $\mu(x)>0$ because $\varrho$ and $\mu$ are continuous.
The minimization with $\overline\gamma(0)$ ensures continuity and well-definedness of $\overline\gamma$ even when $x=0$ or $\mu(x)=0$.
In conclusion, we have shown that $\overline\gamma$ is positive-valued and continuous.
Standard smoothing arguments, e.g., \cite[Theorem 2.2]{hirsch_differential_1976}, can be used to obtain a function $\gamma:\R^n\to\R_{>0}$ that is $\mc C^\infty$ and satisfies $\gamma(x)\leq\overline\gamma(x)$ for all $x\in\R^n$, which also implies that $\frac1\gamma$ is locally Lipschitz.
Then, for any $c>0$, the control law
\begin{align*}
    &v(x,z) := \mc L_{k(x)+z}k(x) - 
    \tfrac1{\gamma(x)}\Big[x^\top\eta_1(x,z) + \eta_2(x,z)z\\
    &~~~~+\tfrac z2\tfrac{\del\gamma(x)}{\del x} f_k(x,z)+\tfrac z4 \text{Tr}\left(g_k(x,z)^\top\tfrac{\del^2\gamma(x)}{\del x^2} g_k(x,z)\right)\\
    &~~~~~~~~~~~\,~~~+cz-\text{Tr}\left(g_k(x,z)^\top \tfrac{\del\gamma(x)}{\del x}^\top\tfrac{\del k(x)}{\del x}g_k(x,z)\right)\\
    &~~~~~~~~~~\,~~~~~+\text{Tr}\left(g_k(x,0)^\top \tfrac{\del k(x)}{\del x}^\top\gamma(x)\tfrac{\del k(x)}{\del x}g_z(x,z)\right)\\
    &~~~+\tfrac12\text{Tr}\left(g_z(x,z)^\top \tfrac{\del k(x)}{\del x}^\top\gamma(x)\tfrac{\del k(x)}{\del x}g_z(x,z)z\right)\Big]\numberthis \label{eq:backstepping:control_law}
\end{align*}
achieves the desired Lyapunov decrease \eqref{eq:backstepping:Lapunov_decrease} because of \eqref{eq:backstepping:calculation}, and because $\gamma(x)\mu(x)|x|^2 \leq \overline\gamma(x)\mu(x)|x|^2 \leq \varrho(x)$ and $\mc L_{k(x)}V(x)\leq-\varrho(x)$ for all $x\in\R^n$.
Furthermore, $v(0,0)=0$, so that the origin is an equilibrium for \eqref{eq:backstepping:system_tilde}.
Finally, $v(x,z)$ is locally Lipschitz because all the functions making up $v(x,z)$ are locally Lipschitz by Assumption \ref{assm:backstepping} and construction of $\gamma$.
We have that the function $\wt V$ is positive definite and radially unbounded, and can therefore be bounded below and above by $\mc K_\infty$ functions as in \eqref{eq:assm:stability:1}.
We have shown that $\widetilde V$ is an SLF for the closed-loop system \eqref{eq:backstepping:system_tilde}.
The Lyapunov theorem \cite[Theorem 5.8]{khasminskii_stochastic_2012} then implies GASp, as summarized in the following theorem.

\begin{thm}
    Suppose Assumptions \ref{assm:backstepping} holds.
    Then the origin of system \eqref{eq:backstepping:system_tilde} under control law \eqref{eq:backstepping:control_law} is GASp.
\end{thm}
\subsection{Exponential case}
If the subsystem \eqref{eq:SBS:1} is exponentially stabilizable, the procedure can be simplified similarly to Sections \ref{s:exponential} and \ref{s:exponential_ISS}.
For that we make the following assumption.
\begin{assm}\label{assm:backstepping:exponential}
    There exists a twice continuously differentiable feedback control law $k:\R^n\to\R$ such that the origin is exponentially 2-stable for \eqref{eq:SBS:1} under $u=k(x)$, and $f,g$ and $k$ have bounded first and second derivatives.
\end{assm}
Under Assumption \ref{assm:backstepping:exponential}, we can apply Proposition \ref{lem:exponential} for $f_k$ and $g_k$ to obtain $V\in\mc C^2(\R^n)$ satisfying \eqref{eq:def:expSLF_1} and \eqref{eq:def:expSLF_3} with $p=2$, and $\mc L_{k(x)}V(x)\leq -a_3|x|^2=:\varrho(x)$ for all $x\in\R^n$.
The procedure in Section \ref{ss:backstepping:design} can then be simplified because $\mu$ defined in \eqref{eq:defn:mu} is bounded, hence $\gamma:=a_3/\sup_{x\in\R^n}\mu(x)>0$ can be chosen constant.
The control law \eqref{eq:backstepping:control_law} then simplifies~to
\begin{align*}
    &v(x,z) := \mc L_{k(x)+z}k(x) - 
    \tfrac1{\gamma}\left[x^\top\eta_1(x,z) + \eta_2(x,z)z+cz\right]\\
    &+\text{Tr}\left(g_z(x,z)^\top \tfrac{\del k(x)}{\del x}^\top\tfrac{\del k(x)}{\del x}\left(g_k(x,0)+\tfrac12g_z(x,z)z\right)\right)\label{eq:backstepping:control_law:exponential}\numberthis
\end{align*}
and achieves exponential stability as follows.
\begin{thm}
    Suppose Assumption \ref{assm:backstepping:exponential} holds. Then the origin of system \eqref{eq:backstepping:system_tilde} under control law \eqref{eq:backstepping:control_law:exponential} is exponentially 2-stable.
\end{thm}

\section{Proofs}\label{s:proofs}
\subsection{Proof of Proposition \ref{thm:main}}\label{s:proof:prop:main}
    Let $\rho\in\mc{PD}$ such that $\varrho-\rho\in\mc{PD}$.
    Let $R>0$, and define the ring-shaped set 
        $\mc R_R:=\{x\in\R^n~|~\min\{1,R\}\leq|x|\leq\max\{1,R\}\},$
    which is compact.
    The function $(x,u)\mapsto \mc L_{u}V(x)$ is continuous on $(\R^n\setminus\{0\})\times\R^m$ because $V\in\mc C^2_0(\R^n)$ and $f$ and $g$ are continuous, and hence uniformly continuous on the compact set $\mc R_R\times\overline\B\subseteq(\R^n\setminus\{0\})\times\R^m$.
    Hence, for any $\varepsilon>0$ there exists $\delta>0$ such that for any $x,x'\in\mc R_R$ and $u,u'\in\overline\B$ with $|x-x'|<\delta$ and $|u-u'|<\delta$, it holds that
        $|\mc L_{u}V(x) - \mc L_{u'}V(x')|<\varepsilon.$
    We are interested in the case where $\varepsilon:=\min_{x\in\mc R_R}(\varrho(x)-\rho(x))$ (which exists and is positive since $\varrho-\rho\in\mc{PD}$ and $\mc R_R$ is compact and excludes $0$) and where $x=x'$ and $u'=0$.
    Then we have that there exists $\delta_R\in(0,1]$ such that for all $x\in\mc R_R$ and $u\in\delta_R\overline\B$,
    \begin{align}
        \mc L_{u}V(x) - \mc L_0V(x)\leq\varrho(x)-\rho(x).\label{eq:proof:thm_main:eps_delta}
    \end{align}
    Let $\delta_R$ be the largest $\delta_R\in(0,1]$ that satisfies
    \eqref{eq:proof:thm_main:eps_delta} for all $x\in\mc R_R$ and $u\in\delta_R\overline\B$; this maximum exists due to continuity and compactness.
    From \eqref{eq:proof:thm_main:eps_delta} and since $\mc L_0 V(x)\leq-\varrho(x)$ for all $x\in\R^n\setminus\{0\}$ by Assumption \ref{assm:stability},
    \begin{align}
        \mc L_{u}V(x) \leq \mc L_0V(x) + \varrho(x) - \rho(x) \leq -\rho(x) \label{eq:thm_main:proof:1.5}
    \end{align}
    for all $x\in\mc R_R$ and $u\in\delta_R\overline\B$. 
    We claim that $R\mapsto\delta_R$ is monotone increasing on $(0,1]$ and decreasing on $[1,\infty)$, 
    which is thanks to $\delta_R$ being chosen maximally.
    Indeed, if $1\leq R_1\leq R_2$, since \eqref{eq:proof:thm_main:eps_delta} holds for all $x\in\mc R_{R_2}$ and $u\in\mc\delta_{R_2}\overline\B$, \eqref{eq:proof:thm_main:eps_delta} also holds for all $x\in\mc R_{R_1}\subseteq\mc R_{R_2}$ and $u\in\mc\delta_{R_2}\overline\B$, that is, $\delta_{R_2}$ would be a viable candidate for defining $\delta_{R_1}$. But, since $\delta_{R_1}$ is chosen maximally, we have $\delta_{R_1}\geq\delta_{R_2}$.
    The monotone increase on $(0,1]$ is proved similarly.

    We will construct a function $\delta\in\mc{PD}$ with $\delta(x)\leq\delta_{|x|}$ for all $x\in\R^n\setminus\{0\}$.
    For that, we first define the function $\widetilde\delta:\R_{\geq0}\to\R_{\geq0}$ as $\widetilde\delta(0):=0$, $\widetilde\delta(s):=\tfrac2s\int_{s/2}^s \delta_t\td t$ for $s\in(0,1]$ and $\widetilde\delta(s):=\tfrac1s\int_s^{2s} \delta_t\td t$ for $s\in[1,\infty)$.
    By the aforementioned monotonicity of $R\mapsto\delta_R$ we have for $s\in(0,1]$ that $0<\delta_{s/2}\leq \widetilde\delta(s)\leq \delta_s$ holds, and for $s\in(1,\infty)$ that $0<\delta_{2s}\leq\widetilde\delta(s)\leq\delta_s$ holds.
    Furthermore, $\widetilde\delta$ is continuous on $(0,1)$ and $(1,\infty)$ thanks to the integrations (and using that $\delta_t$ is bounded by 1).
    To achieve continuity at $0$ and $1$, we define $\overline\delta:\R_{\geq0}\to\R_{\geq0}$ as $\overline\delta(0):=0$ and
    $\overline\delta(s) := \min\{s, \widetilde\delta(s), \widetilde\delta(1), \lim_{t\searrow1}\widetilde\delta(t)\}$ for $s>0$.
    Also, $0<\overline\delta(s)\leq\delta_s$ for all $s\in(0,\infty)$.
    Finally, define $\delta:\R^n\to\R_{\geq0}$ as $\delta(0):=0$ and $\delta(x):=\overline\delta(|x|)$.
    Then $\delta\in\mc{PD}$ and $\delta(x)\leq\delta_{|x|}$ for all $x\in\R^n\setminus\{0\}$.
    
    We now show \eqref{eq:lem:main}.
    Let $x\in\R^n\setminus\{0\}$ and $u\in\R^m$ with $|u|\leq\delta(x)$.
    Then, because $|u|\leq\delta(x)\leq\delta_{|x|}$ and $x\in\mc R_{|x|}$, inequality  \eqref{eq:proof:thm_main:eps_delta} and thus \eqref{eq:thm_main:proof:1.5} holds, which completes the proof.

\subsection{Proof of Theorem \ref{thm:stronger}}\label{s:proof:thm:stronger}
We follow the proof of Proposition \ref{thm:main} with $\rho:=\varrho/2$ and show $\limsup_{x\to0}\delta(x)/|x|>0$ for the $\delta$ constructed there.
Because the functions $f_u, g_u^i$ and $g_x^i$ defined in Section \ref{s:stronger} are continuous, there exists $B>0$ such that
\begin{align}
    \max\left\{\n{f_u(x,u)}, \n{g_u^i(x,u)}, \n{g_x^i(x)}\right\}\leq B
    \label{eq:proof:lem_exponential:bounds}
\end{align}
for any $x\in\overline\B, u\in\overline\B$ and $i\in\{1,\dots,p\}$.
Furthermore, by \eqref{eq:assm:stronger:2} there exists $C>0$ such that $\left|\tfrac{\del^iV(x)}{\del x^i}\right|\tfrac{|x|^i}{\varrho(x)}\leq C$ for $i=1,2$ and any $x\neq0$ in a neighbourhood of $0$, and because $\frac{\del^iV}{\del x^i}$ are continuous and $\varrho\in\mc{PD}$, $C$ can even be chosen such that this holds for all $x\in\overline\B\setminus\{0\}$.
Then, for any $x\in \overline\B\setminus\{0\}$ and $u\in\overline\B$, by \eqref{eq:LuV_stronger} and choice of $B$ and $C$,
\begin{align*}
    &\mc L_uV(x) - \mc L_0V(x)\\
    &\leq{C\varrho(x)}|x|^{-1}B|u| + (n|x|+ p|u|)B{C\varrho(x)}|x|^{-2}B|u|\\
    &\leq \varrho(x)\left((BC+nB^2C){|u|}{|x|}^{-1} + pB^2C{|u|^2}{|x|}^{-2}\right).\numberthis
\end{align*}
Choose $b:=\tfrac14\min\{(BC+nB^2C)^{-1}, (pB^2C)^{-1/2},1\}$.
Then, for any $x\in\overline\B\setminus\{0\}$ and $u\in\R^m$ with $|u|\leq b|x|$, we have $\mc L_uV(x) - \mc L_0V(x) \leq \varrho(x)/2$.
    Hence, following the proof of Proposition \ref{thm:main} with $\rho := \varrho/2$, we get that \eqref{eq:proof:thm_main:eps_delta} is satisfied for all $(x,u)\in\overline\B\times\R^m$ with $|u|\leq b|x|$.
    Since $\delta_R$ in the proof of Theorem \ref{thm:1.robustness} is chosen maximally to satisfy \eqref{eq:proof:thm_main:eps_delta}, we get $\delta_R\geq bR$ for all $R\in(0,1]$.
    Further following the proof of Theorem \ref{thm:1.robustness}, $\widetilde\delta(s)\geq \delta_{s/2}\geq \tfrac b2s$ for all $s\in(0,1]$, and $\overline\delta(s) \geq \min\{1,b/2\}s$ for all $s\in\left(0,\min\left\{\widetilde\delta(1),\lim_{t\searrow1}\widetilde\delta(t)\right\}\right]$.
    Hence,
    $\liminf_{x\to0}\delta(x)/|x| = \liminf_{s\searrow0}\overline\delta(s)/s \geq \min\left\{1,b/2\right\} > 0,$
    completing the proof.

\subsection{Proof of Proposition \ref{lem:exponential}}\label{proof:prop:exponential}
    By \cite[Theorem 5.12]{khasminskii_stochastic_2012}\footnote{\cite[Theorem 5.12]{khasminskii_stochastic_2012} considers time-varying systems, but the construction in the proof yields a time-invariant Lyapunov function for time-invariant systems.}, there exist $V\in\mc C_0^2(\R^n)$ and $a_1,a_2,a_3,a_4>0$ satisfying \eqref{eq:def:expSLF_1}, \eqref{eq:def:expSLF_2} and \eqref{eq:def:expSLF_3}
for all $x\in\R^n\setminus\{0\}$.
    Because $f$ and $g$ have bounded derivatives by Assumption \ref{assm:exp_stability}, the functions $f_u, g_u, g_x$ defined in Section \ref{s:stronger} are also bounded, that is, there exists $B>0$ such that \eqref{eq:proof:lem_exponential:bounds} holds
for any $(x,u)\in\R^n\times\R^m$ and $i\in\{1,\dots,p\}$.
Using \eqref{eq:def:expSLF_3} for $i=2$, the function $V_x(x) = \int_0^1\tfrac{\del^2 V(sx)}{\del x^2}ds$ defined in Section \ref{s:stronger} is bounded for any $x\in\R^n\setminus\{0\}$ by
    $\n{V_x(x)} 
    \leq \int_0^1 a_4|sx|^{p-2}ds \leq \tfrac{a_4}{p-1}|x|^{p-2}\label{eq:proof:lem_exponential:Vx}$
since $p>1$. With this, \eqref{eq:LuV_stronger} and \eqref{eq:proof:lem_exponential:bounds}, it follows that for any $(x,u)\in(\R^n\setminus\{0\})\times\R^m$,
\begin{align}
    \mc L_uV(x) &\leq \mc L_0V(x) + |x|\tfrac{a_4}{p-1}|x|^{p-2}B|u|\nonumber\\
    &~~~+ n|x|Ba_4|x|^{p-2}B|u| + p|u|Ba_4|x|^{p-2}B|u|\nonumber\\
    &\leq \mc L_uV(x) + C(|x|^{p-1}|u| + |x|^{p-2}|u|^2)
\end{align}
with $C := \max\{\tfrac{a_4}{p-1}B + nB^2a_4, pB^2a_4\}$.

\section{Conclusion}\label{s:conclusion}
We analyzed robustness and input-to-state stability of SDEs under various sets of assumptions.
In the most general case we showed GASp under a state-dependent perturbation bound vanishing at 0, as well as SISS under a state-dependent perturbation bound positive everywhere.
We gave conditions under which perturbations can be up to proportional locally.
Furthermore, we showed that nominal $p$-stability is maintained under a proportional perturbation bound and implies exponential stochastic ISS. 
Finally, we illustrated some of the techniques for backstepping.

\section*{References}
\vspace{-.7cm}
\bibliographystyle{IEEEtranS}
\bibliography{references}

\end{document}